\documentclass[a4paper,11pt]{article}

\usepackage{authblk}
\usepackage[utf8]{inputenc}
\usepackage{amsthm}

\usepackage{amsmath,amssymb,enumitem}
\usepackage[linesnumbered, ruled, vlined]{algorithm2e} 

\newtheorem{theorem}{Theorem}

\newtheorem{lemma}[theorem]{Lemma}
\newtheorem{corollary}[theorem]{Corollary}
\newtheorem{definition}[theorem]{Definition}

\usepackage{tikz, pgfplots}
\usetikzlibrary{patterns}
\usetikzlibrary{shapes.geometric, calc}

\usepackage[margin=1in]{geometry}
\usepackage{hyperref}
\usepackage{color}
\definecolor{darkgreen}{RGB}{0,100,0}
\definecolor{firebrick}{RGB}{178,34,34}

\newcommand{\R}{{\mathbb{R}}}

\newcommand{\pth}[1]{\ensuremath{\left(#1\right)}}

\newcommand{\calS}{{\cal S}}

\newcommand{\calG}{{\cal G}}

\newcommand{\disc}{{\rm disc}}

\newcommand{\e}{{\varepsilon}}
\newcommand{\calD}{{\mathcal{D}}}
\newcommand{\calO}{{\mathcal{O}}}
\newcommand{\calP}{{\mathcal{P}}}
\newcommand{\calR}{{\mathcal{R}}}

\newcommand*\samethanks[1][\value{footnote}]{\footnotemark[#1]}

\title{Sparse Geometric Set Systems and the Beck-Fiala Conjecture}
\date{}

\author[1]{Kunal Dutta
\thanks{
Supported by the Polish NCN-SONATA Grant no. 2019/35/D/ST6/04525 
\emph{(Probabilistic tools for high-dimensional geometric inference, topological data analysis and large-scale networks)}.}
}

\author[2]{Arijit Ghosh 
\samethanks[1]
}

\affil[1]{Faculty of Mathematics, Informatics, and Mechanics, University of Warsaw, Poland. email: K.Dutta@mimuw.edu.pl}
\affil[2]{Indian Statistical Institute, Kolkata, India email: arijitiitkgpster@gmail.com}

\begin{document}

\maketitle

\begin{abstract}
   We investigate the combinatorial discrepancy of geometric set systems having bounded shallow cell complexity in the \emph{Beck-Fiala} setting, where each 
   point belongs to at most $t$ ranges. For set systems with shallow cell complexity $\psi(m,k)=g(m)k^{c}$, where $(i)$ $g(m) = o(m^{\e})$ 
   for any $\e\in (0,1],$ $(ii)$ $\psi$ is  non-decreasing in $m$, and $(iii)$ $c>0$ is independent of $m$ and $k$, we get a discrepancy bound of 
   \[ O\pth{\sqrt{\pth{\log n+\pth{t^{c}g(n)}^{\frac{1}{1+c}}}\log n}}.\]
   For $t=\omega(\log^2 n)$, in several cases, such as for set systems of points and half-planes / disks / pseudo-disks in $\R^2$, points and 
   orthants in $\R^3$ etc., 
   these bounds are $o(\sqrt{t})$, which verifies (and improves upon) the conjectured bound of Beck and Fiala~\emph{(Disc. Appl. Math., 1981)}. \\

   Our bounds are obtained by showing the existence of \emph{matchings with low crossing number}, using the multiplicative weights update method of Welzl 
   \emph{(SoCG, 1988)}, together 
   with the recent bound of Mustafa \emph{(Disc. Comp. Geom., 2015)} on \emph{shallow packings} of set systems in terms of their shallow cell complexity. 
   For set systems of shallow cell complexity $\psi(m,k)=m^{c_1}g(m)k^{c}$, we obtain matchings with crossing number at most 
   \[ O\pth{\pth{n^{c_1}g(n)t^{c}}^{\frac{1}{1+c_1+c}}}.\]
   These are of independent interest.
\end{abstract}

\section{Introduction}
\label{sec:intro}

   Given a set system $(X,\calR)$ with ground set $X$ and a collection $\calR\subset 2^X$ of subsets of $X$, the \emph{combinatorial 
discrepancy} of the system is given by 
\[ 
    \disc(\calR) := 
    \min_{\chi: X \rightarrow \{\pm 1\}}\max_{S\in \calR}\left|\sum_{i\in S} \chi(i) \right|,
\]
that is, the maximum imbalance over all sets in $\calR$, minimized over all possible $2$-partitions of $X$.
Using the incidence matrix $A$ of the system $(X,\calR)$, an equivalent linear algebraic formulation -- the \emph{vector balancing} 
problem can be stated as
\[ \disc(A) := \min_{x\in \{-1,1\}^n} \|Ax\|_{\infty}.\]
For a given class of set systems, the problem of \emph{discrepancy minimization} seeks to establish bounds on the 
discrepancy of systems in the class, as a function of the size of the ground set, or other parameters. Besides its 
inherent interest, discrepancy minimization has many applications in several areas of mathematics as well as computer 
science and related subjects~\cite{Chaz01,Matousekdiscbook09}.  

\paragraph*{Beck-Fiala and Komlos Conjectures}
In their seminal work which initiated the use of the term ``discrepancy" as well as the study of the discrepancy of set systems, 
Beck and Fiala~\cite{Beck-Fiala} showed that any set system where each element belongs to at most $t$ sets, has discrepancy at most $2t-1$. 
They conjectured that this bound could be improved to $O\pth{\sqrt{t}}$. In linear algebraic terms, the closely related generalization 
known as the \emph{Komlos conjecture}, states that for a real matrix with $\ell_2$ norm of each column at most $1$, the discrepancy is $O(1)$. 
Despite significant progress in partial results and 
understanding of discrepancy theory in the last $4$ decades, these tantalizing conjectures have remained open. 

\paragraph*{Early and Recent Progress}
Since the seminal results of Beck and Fiala, much progress has been made toward resolving 
the Beck-Fiala and Komlos conjectures. Early breakthroughs in the area include those of Spencer
~\cite{Sp85} and Gluskin~\cite{Gluskin89} who showed that the Komlos conjecture holds for $t= \Theta(n)$.
Later Srinivasan~\cite{SrinivasanSODA97} and Banaszczyk~\cite{Banas} proved bounds for general $t$. The latter's 
bound of $5\sqrt{t\log n}$ remains the best current result on the Komlos conjecture. \\

The above bounds were non-constructive and it was even conjectured that efficient algorithms to 
construct colorings with the promised discrepancy bounds, did not exist. However in a major breakthrough, Bansal~\cite{Ba10}
gave an efficient algorithm to construct a coloring matching Spencer's discrepancy bounds for $|\calR|= \Theta(n)$. 
This has been followed by a large number of recent results, both algorithmic and non-algorithmic, for upper and lower 
bounds, as well as many applications and generalizations, see~\cite{LM12,BansalDG16,BansalG16a,DadushGLN16,LevyRR16,https://doi.org/10.1112/S0025579317000250,10.1145/3188745.3188850}.

\paragraph*{Geometric Set Systems}
The discrepancy of set systems with a geometric interpretation, such as for example of points and half-planes in $\R^2$ and $\R^{3}$ or points and disks in $\R^{2}$,
has been of significant theoretical and practical importance, as it relates to several areas like computational geometry, statistical and 
machine learning, algorithmic analysis, database theory, etc.~\cite{Matousekdiscbook09,Larsen14,DBLP:conf/stacs/Larsen19}. 
Questions such as Tu\'snady's problem have been studied since the 
early years of discrepancy theory~\cite{Beck81c,DBLP:journals/dm/Beck89}. Other notable results include those of points and 
projective planes (Spencer~\cite{SPENCER1988213}), primal systems 
of bounded VC dimension (Matou\v{s}ek~\cite{M99}), systems of bounded dual VC dimension, etc.(Matou\v{s}ek-Welzl-Wernisch~\cite{MatousekWW93-discrepancy-approx-VC}). \\

More recently Ezra~\cite{E14} proposed the notion of size-sensitive discrepancy, and further asked about the discrepancy of geometric set 
systems with bounded degree. Using the packing and chaining framework of Dutta, Ezra and Ghosh~\cite{DuttaEG15-shallow-packings}, together with the 
partial coloring based algorithm of Lovett-Meka (or a similar algorithm), it is comparatively 
straightforward to obtain bounds of $O\pth{t^{1/4}\cdot\log n}$ for points and half-planes in $\R^2$. To some extent, this setting has been 
considered in the unpublished work~\cite{DG14}, 
but the bounds have an extra $\sqrt{\log\log t}\log n$ factor (as in the above example) which stems from using 
the partial coloring approach and a suboptimal assignment of parameters. \footnote{Moreover these bounds do not appear 
in the published version.} To the best of our knowledge,
other than the above-mentioned unpublished drafts, there are no published or publicly available works which consider the question of 
geometric set systems with bounded degree of elements. However interesting questions remain to be asked for such systems.
For example it seemed interesting to ask if the $O\pth{t^{1/4}\log n}$ bound for points and half-planes in the Beck-Fiala setting, 
can be improved to $O\pth{t^{1/4}\sqrt{\log n}}$ or better, using (for instance) the algorithmic frameworks of Bansal and others.

\paragraph*{Union and Shallow Cell Complexity}
Set systems with low \emph{union complexity} or \emph{shallow cell complexity} constitute a fairly general class and include several important cases such as points and 
half-spaces in $\R^2$ and $\R^3$, points and (pseudo)disks in $\R^2$ and several others (see, for example~\cite[Chapter 4]{mustafa2022samplingbook}). Moreover, 
since dual systems of low union complexity also have low 
shallow cell complexity of the corresponding primal systems (e.g.~\cite{MV17}), they have been studied in a series of works e.g.~\cite{ClarksonV07,V10,CGKS12,DBLP:journals/jocg/BansalP16,MDG17,DuttaGJM19,DuttaGM22}
where better bounds have been obtained for them for several structures of interest, such as unweighted and weighted epsilon nets, epsilon approximations, 
relative approximations, $\epsilon$-brackets, combinatorial Macbeath regions, etc.

\subsection*{Our Contribution}

We investigate the discrepancy of geometric set systems with linear or near-linear shallow cell complexity, and each element belonging 
to at most $t$ sets, where $t$ is a given parameter independent of $n$. Our discrepancy bounds are stated below. 
\begin{theorem}
\label{thm:discr-bd-crossn-degt}
    Let $(X,\calR)$ be a set system with shallow cell complexity $\psi(m,k)=g(m)k^c$, where $g(m) = o(m^{\e})$ for every $\e\in (0,1)$, is a non-decreasing 
    function of $m$ and $c>0$ is a constant independent of $n$ and $k$, and each element of $X$ belongs to at most $t$ sets of 
    $\calR$. Then the discrepancy of $(X,\calR)$ is at most 
   \[ O\pth{\sqrt{\pth{\log n+\pth{t^{c}g(n)}^{\frac{1}{1+c}}}\log n}}.\]
    In particular, for the following families of set systems with degree at most $t$, the discrepancy is bounded as below. 
    \begin{enumerate}
        \item Points and half-planes in $\R^2$: $O\pth{\sqrt{(\log n+t^{1/2})\log n}}$.
        \item Points and homothets of a convex body in $\R^2$: $O\pth{\sqrt{(\log n+t^{1/2})\log n}}$.
        \item Points and disks in $\R^2$: $O\pth{\sqrt{(\log n+t^{1/2})\log n}}$.
        \item Points and pseudodisks in $\R^2$: $O\pth{\sqrt{(\log n+ t^{1/2})\log n}}$.
        \item Points and $\alpha$-fat triangles in $\R^2$: $O\pth{\sqrt{(\log n+(t\cdot\log^*n)^{1/2})\log n}}$.
        \item Points and objects with linear union complexity: $O\pth{\sqrt{(\log n+t^{1/2})\log n}}$.
        \item Points and locally $\gamma$-fat semi-algebraic objects in $\R^2$ with bounded description complexity: 
              $O\pth{\sqrt{(\log n+t^{1/2}\cdot 2^{O(\log^* n)})\log n}}$.
        \item Points and half-spaces in $\R^3$: $O\pth{\sqrt{(\log n + t^{2/3})\log n}}$.
        \item Points and orthants in $\R^2$ and $\R^3$:  $O\pth{\sqrt{(\log n+t^{1/2})\log n}}$.
    \end{enumerate}
\end{theorem}

These bounds match or improve the known, as well as the conjectured general bounds for systems of bounded degree. 
For instance, for set systems of points and half-planes in $\R^2$, we get the bound of 
   \[ O\pth{\pth{\pth{\log n+\sqrt{t}}\log n}^{1/2}},\] 
   which is $O\pth{\sqrt{t}}$ for $t = \Omega(\log n)$.
For $t= \omega\pth{\log^2 n}$, we have $\sqrt{\log n} = o(t^{1/4})$, so the $O\pth{t^{1/4}\sqrt{\log n}}$ bound improves upon the conjectured 
$O\pth{\sqrt{t}}$ bound of Beck and Fiala. Thus, they give much better bounds even without using the random walk based framework and intricate analysis of 
many recent discrepancy minimization algorithms e.g.~\cite{LM12,BansalG16a,Rothvoss16}. \\

More generally, we get the following corollary of Theorem~\ref{thm:discr-bd-crossn-degt}.
\begin{corollary}
    The Beck-Fiala conjecture holds for the following set systems having maximum degree $t$, with the given constraints on $t$.   
    \begin{enumerate}
        \item Points and half-planes in $\R^2$: $t=\Omega\pth{\log^2 n}$.
        \item Points and homothets of a convex body in $\R^2$: $t=\Omega\pth{\log^2 n}$.
        \item Points and disks in $\R^2$: $t=\Omega\pth{\log^2 n}$.
        \item Points and pseudodisks in $\R^2$: $t=\Omega\pth{\log^2 n}$.
        \item Points and $\alpha$-fat triangles in $\R^2$: $t=\Omega\pth{\log^2 n\log^*n}$.
        \item Points and objects with linear union complexity: $t=\Omega\pth{\log^2 n}$.
        \item Points and locally $\gamma$-fat semi-algebraic objects in $\R^2$ with bounded description complexity: 
              $t=\Omega\pth{2^{O(\log^* n)})\log^2 n}$.
        \item Points and half-spaces in $\R^3$: $t=\Omega\pth{\log^3 n}$.
        \item Points and orthants in $\R^2$ and $\R^3$:  $t=\Omega\pth{\log^2 n}$.
    \end{enumerate}
\end{corollary}

Our techniques are based upon showing the existence of \emph{matchings with low crossing number} using the well-known multiplicative weights update 
method of Welzl~\cite{10.1145/73393.73397}, together with the shallow packing bound of Mustafa~\cite{M-shallow-packing-2016} and Dutta, Ezra and Ghosh~\cite{DuttaEG15-shallow-packings}. 
These crossing number bounds, which are of independent interest, are as below.

\begin{theorem}
\label{thm:low-cross-match-dual}
Let $(X,\calR)$ be an $n$-point set system, with $|\calR|\geq n$, generated by intersections of $X$ with a class of objects having dual shatter 
dimension at most $d$ and shallow cell complexity $\psi(m,k) = m^{c_1}\cdot g(m)k^{c}$, where
$c_1,c \geq 0$ are independent of $m$ and $k$, and $g(m)$ is a non-decreasing function of $m$ such that $g(m) = o(m^{\e})$ for every $\e \in (0,1)$. 
Further, let $(X,\calR)$ have the property that each point of $X$ belongs to at most $t$ sets.
Then there exists a matching of the points of $X$, with crossing number at most 
\[ O\pth{\pth{n^{c_1}t^{c} g(n)}^{\frac{1}{1+c_1+c}}}.\]
In particular, if $c_1=0$, then the crossing number 
is at most 
\[ O\pth{\pth{t^{c} g(n)}^{\frac{1}{1+c}}}.\]
\end{theorem}

\begin{corollary}
\label{cor:low-cross-match-list}
    For the following set systems $(X,\calR)$ with each element of $X$ belonging to at most $t$ members of $\calR$, 
there exist matchings with the following bounds on the crossing number. 
    \begin{enumerate}
        \item Points and half-planes in $\R^2$: $O\pth{\log n+\sqrt{t}}$.
        \item Points and homothets of a convex body in $\R^2$: $O\pth{\log n+\sqrt{t}}$.
        \item Points and disks in $\R^2$: $O\pth{\log n+\sqrt{t}}$.
        \item Points and pseudodisks in $\R^2$: $O\pth{\log n+\sqrt{t}}$.
        \item Points and $\alpha$-fat triangles in $\R^2$: $O\pth{\log n+\sqrt{t\cdot \log^* n}}$.
        \item Points and objects with linear union complexity: $O\pth{\log n+\sqrt{t}}$.
        \item Points and locally $\gamma$-fat semi-algebraic objects in $\R^2$ with bounded 
        description complexity: $O\pth{\log n+\sqrt{t\cdot 2^{O(\log^* n)}}}$.
        \item Points and half-spaces in $\R^3$: $O\pth{\log n+t^{2/3}}$.
        \item Points and orthants in $\R^2$ and $\R^3$: $O\pth{\log n+t^{1/2}}$.
    \end{enumerate}
\end{corollary}

The proofs of our results are in Section~\ref{sec:main}.
While our discrepancy and crossing number bounds are not difficult to obtain, we believe they are important, as 
they demonstrate cases of non-random (and non-smoothed) natural set systems which verify or improve upon the Beck-Fiala conjecture. 
Further, these bounds are much stronger than the bounds for general set systems of bounded VC dimension or 
shallow cell complexity (without the degree bound), which are known to be tight~\cite{matousek_1997}. 
It should be noted that merely using low shallow complexity cannot improve the bound on the crossing number, 
as the tightness results in~\cite{matousek_1997} hold for systems in $2$ and $3$ dimensions. 
Matchings (or spanning paths or trees) with low crossing number are of wide interest, as they are used in a number of 
applications~\cite{Matousekdiscbook09,DBLP:conf/compgeom/CsikosM21}.

\section{Preliminaries}
\label{sec:prelims}

The \emph{projection} of a set system $(X,\calR)$ on to a subset $Y\subset X$ of the ground set is denoted by     
    $\calR_{\mid Y} := \left\{ R \cap Y \; \mid \; R \in \calR \right\}$. 

A \emph{primal} set system has a ground set $P$ of points in $\R^d$ and subsets given by all possible intersections 
with a class $\calG$ of geometric objects in $\R^d$, that is the set system $(P,\calG_{\mid P})$. A \emph{dual} set system given 
a (finite) collection $\calG$ of geometric objects is given by $(\calG,\calG_{\mid \R^d})$, that is, equivalence classes of points of $\R^d$ under 
intersection with objects in $\calG$. \\

The notion of \emph{shallow cell complexity} has found many applications
in Computational Geometry, including improved bounds on $\e$-nets and related structures (see e.g.~\cite{MV17}).
        \begin{definition}[Shallow-cell Complexity]
        \label{def:shall-cell-complx}
         A set system $(X, \calR)$ has \emph{shallow-cell complexity}
         $\psi(\cdot, \cdot)$ if for any $Y \subseteq X$, 
         the number of subsets in $\calR|_Y$ of size $l$
         is at most $ |Y| \cdot \psi(|Y|, l)$.
        \end{definition}

In this paper, we shall throughout work with set systems whose shallow cell complexity 
is given by 
\begin{eqnarray}
   \psi(l,k) &=& l^{c_1}g(l)k^c, \label{eqn:scc-def-assump} 
\end{eqnarray}
where $c_1,c>0$ are constants independent of $l,k$, and $g(l) = o(l^{\e})$ for every $\e\in (0,1]$.

        \begin{definition}[Union Complexity]
        \label{def:un-complx}
         A set of $m$ geometric objects $\calO$ has \emph{union complexity}
         $f(m)$ if the number of faces of all dimensions in the union of the members 
         of $\calO$, is at most $f(m)$.
        \end{definition}
It is known (see e.g.~\cite{MV17}) that set systems formed by intersections of points with objects having union complexity
$f(.)$, have shallow cell complexity at most $f(m/k)\cdot k^2$. \\

A central lemma for set systems of bounded shallow cell complexity, is the Shallow Packing Lemma 
of~\cite{DuttaEG15-shallow-packings,M-shallow-packing-2016}. To state the lemma, first we need 
to define \emph{shallow packings}.
\begin{definition}[Shallow Packing]
    Let $(X, \calR)$ be a set system, and $\delta$ and $k$ be positive integers. 
        A subset of ranges $\calP \subseteq \calR$ is a \emph{$k$-shallow $\delta$-packing} or 
\emph{$(k,\delta)$-packing}, if any pair of ranges in $\calP$ have symmetric difference 
        greater than $\delta$, and for all $R \in \calP$ we have $|R\cap X| \leq k$.
\end{definition}


Now follows the Shallow Packing Lemma of Mustafa~\cite{M-shallow-packing-2016}, which gives optimal bounds for 
shallow packings of set systems, in terms of their shallow cell complexity.
\begin{theorem}[Shallow Packing Lemma, Mustafa~\cite{M-shallow-packing-2016}]
\label{thm:shallow-packing-lemma}
    Let $(X, \calR)$ be a set system with $|X| = n$ and shallow cell complexity  $\varphi_{\calR}$. If 
the VC dimension of $(X, \calR)$ is at most $d_{0}$, and $(X, \calR)$ is a $k$-shallow $\delta$-packing then 
    \[ 
        \frac{24d_{0}n}{\delta} \cdot \varphi_{\calR}\left( \frac{4d_{0}n}{\delta}, \frac{12d_{0}k}{\delta}\right) \leq C_{d_0}\frac{n}{\delta}\psi\pth{\frac{n}{\delta},\frac{k}{\delta}},
    \]
where $C_{d_0}$ is a constant independent of $n$ and $\delta$.
\end{theorem}

Finally, we define the notions of \emph{crossing number} and partitions with low crossing number, which are central to this paper.
\begin{definition}[Crossing Number; Partitions with Low Crossing Number]
   A range $S \in \calR$ is \emph{crossed} by a pair of elements of $X$, $\{x,y\}$ if $|\{x,y\}\cap X|=1$. Given a partition of the universe $X$ into 
pairs $\Pi=(\{x_i,y_i\})_{i=1}^{n/2}$ with $\{x_i,y_i\}\cap \{x_j,y_j\} = \emptyset$ for $i\neq j$, and $\bigcup_{i\geq 1}\{x_i,y_i\} = X$, the \emph{crossing 
number} of $\Pi$ is the maximum number of pairs of $\Pi$ which cross any given range of $\calR$. 
\end{definition}

\section{Well-Behaved Dual Systems}
\label{sec:main}

In this section, we prove our theorems and corollaries. We begin with the proof of Theorem~\ref{thm:discr-bd-crossn-degt}. The proof follows immediately from the 
statements of Theorem~\ref{thm:low-cross-match-dual} and Corollary~\ref{cor:low-cross-match-list}, with the following additional lemma,
which gives a bound on the discrepancy of an arbitrary set system, in terms of the \emph{crossing number} of a perfect matching of the elements.
The first statement in the lemma is an easy observation, while the second statement is well-known (see e.g. Chapter 5~\cite{Matousekdiscbook09}). 
\begin{lemma}
\label{l:discr-bd-crossn-degt}
    Let $(X,\calR)$ be a set system, $|X|$ even, with degree bounded by $t$ and having a matching with crossing number at most $N$. Then the discrepancy of $(X,\calR)$ is at most 
    \[O
        \left(\min\left\{N,\sqrt{N\cdot\log |\calR|}\right\}\right) = O
        \left(\sqrt{N\cdot\log |\calR|}\right).
    \] 
\end{lemma}

\begin{proof}
   It is easy to see that the discrepancy is at most $N$: just take an arbitrary coloring where each vertex of matched pair of vertices, is given opposite colors. The discrepancy 
of a set $S\in \calR$ is at most the number of pairs that have exactly one vertex in $S$. This is at most the crossing number, i.e. $N$.

   To see that the discrepancy is at most $O\pth{\sqrt{N\log |\calR|}}$, consider a random coloring which ensure that in every matched pair, opposite colors are assigned to the vertices 
in the pair. Now for any set $S\in \calR$, the expected discrepancy of $S$ is zero, and the variance of the discrepancy of $S$ is at most the number of pairs which cross $S$, i.e. at most the
crossing number $N$. Now applying Chernoff bounds, allowing for a maximum deviation of $O\pth{\sqrt{N\log |\calR|}}$, and taking a union bound over all sets of $\calR$, we 
get the statement of the lemma.
\end{proof}

For the proof of Theorem~\ref{thm:low-cross-match-dual}, we shall need the following lemma, which guarantees the existence of a pair 
of points which cross a small number of ranges. 
\begin{lemma}[Short Edge Lemma]
\label{l:short-edge-uc}
    Assume the setting of Theorem~\ref{thm:low-cross-match-dual}. Then for any set or multiset of ranges $T \subseteq \calR$, 
there exist points $x,y \in X$, such that the edge $\{x,y\}$ is crossed by at most $\delta$ ranges, where $\delta$ is the 
maximum value satisfying  
       \begin{eqnarray*}
            n &\leq&  C_{d_{0}}\frac{|T|}{\delta}\cdot\psi\pth{\frac{|T|}{\delta},\frac{t}{\delta}}, 
       \end{eqnarray*}
where $C_{d_{0}}$ is a constant that depends only on $d_{0}$.
\end{lemma}

\begin{proof}
   For each point $x \in X$, define the set of ranges $D_x := \{S\in T\;:\; x\in S\}$.
Let $\calD := \{D_x\;:\; x\in X\}$ be the collection of sets of ranges so formed. 
Then for any pair of points $x,y \in X$, the symmetric difference $D_x\Delta D_y$ is the set of 
ranges which are crossed by the pair $\{x,y\}$.
The shallow cell complexity of the system $\calD$ is at most the union complexity of the dual 
set system given by $T$. Let 
$\delta$ denote the minimum symmetric difference between any pair of sets in $\calD$.
If there exists a pair $\{x,y\}$ such that $D_x = D_y$, then we are done, since $\delta=0$ for this 
pair. Otherwise, observe that by definition, $\calD$ 
is a $\delta$-packing. Since every element of $(X,\calR)$ belongs to at most $t$ ranges, the 
maximum size of any set in $\calD$ is $t$. Thus, we can apply the Shallow Packing Lemma~\ref{thm:shallow-packing-lemma} 
to bound the size of the $\delta$-packing. We get  
       \begin{eqnarray*}
            n &=& |\calD| \;\;\leq\;\; C_{d_0}\frac{|T|}{\delta}\cdot\psi\pth{\frac{|T|}{\delta},\frac{t}{\delta}}. 
       \end{eqnarray*}  
where $C_{d_{0}}$ is a constant that depends only on $d_{0}$.
\end{proof}

We will now give the proof of Theorem~\ref{thm:low-cross-match-dual}.
\begin{theorem}
[Restatement of Theorem~\ref{thm:low-cross-match-dual}]
\label{thm:low-cross-match-dual-restatement}
Let $(X,\calR)$ be an $n$-point set system, with $|\calR|\geq n$, generated by intersections of $X$ with a class of objects having dual shatter 
dimension at most $d$ and shallow cell complexity $\psi(m,k) = m^{c_1}\cdot g(m)k^{c}$, where
$c_1,c \geq 0$ are independent of $m$ and $k$, and $g(m)$ is a non-decreasing function of $m$ such that $g(m) = o(m^{\e})$ for every $\e \in (0,1)$. 
Further, let $(X,\calR)$ have the property that each point of $X$ belongs to at most $t$ sets.
Then there exists a matching of the points of $X$, with crossing number at most 
\[ O\pth{\pth{n^{c_1}t^{c} g(n)}^{\frac{1}{1+c_1+c}}}.\]
In particular, if $c_1=0$, then the crossing number 
is at most 
\[ O\pth{\pth{t^{c} g(n)}^{\frac{1}{1+c}}}.\]
\end{theorem}
\begin{proof}
   We shall construct the matching by picking pairs of points sequentially, in a greedy manner. Each time, we shall
   choose a pair of points that crosses the least number of ranges. However, instead of looking at the number of ranges crossed
   by a given pair, we shall do a weighted counting, in which we shall assign a weight to each range and look to choose the pair
   that minimizes the \emph{weight} of the ranges that it crosses. The weight of a range $S\in \calR$ will be a function of the number of pairs
   already selected, which cross $S$. This will force the selection of edges which do not cross the ranges that have already been 
   crossed several times, thus our strategy will tend to favour low crossing number with respect to the already selected edges. \\

   Initially, we assign each range a weight of $1$. After the $i$-th pair has been chosen, for a given range $S\in \calR$ let $c_S(i)$ denote the number 
   of already selected pairs, which cross $i$. Then the weight of $S$ is given by $w_i(S) = 2^{c_S(i)}$. The weight of any subset of ranges $\calS\subseteq \calR$ is the sum 
   of the weights of the ranges in $\calS$. 
   At each step, we greedily pick a pair that crosses the least number of weighted ranges of $\calR$, where the weights have been updated at the end 
   of the previous step. At the end of the $i$-th step, let $T_i$ denote the multiset obtained $\calR$ by replacing each range $S\in \calR$ 
   with $w_i(S)$ copies of $S$. Note that $w_i(\calR) = |T_i|$, where the cardinality of $T_i$ refers to the weighted cardinality.

   By the Short Edge Lemma~\ref{l:short-edge-uc}, we have 
\begin{eqnarray*}
   n &\leq& \frac{|T_i|^{1+c_1}}{\delta^{1+c_1+c}}g(|T_i|/\delta)t^{c}\\
   &\leq&    \frac{w_i(\calR)^{1+c_1}}{\delta^{1+c_1+c}}g(w_i(\calR))t^{c}.
\end{eqnarray*}

Therefore, there exists a pair of points in 
\[
    X\setminus \{x_1,y_1,x_2,y_2,\ldots,x_i,y_i\},
\]
which cross at most $\delta$ ranges, 
where $\delta$ is the maximum value satisfying 
\[ \delta \leq \pth{\frac{|T_i|^{1+c_1}}{n}g(|T|)t^{c}}^{1/(1+c_1+c)}.\]
We choose this pair $\{x_{i+1},y_{i+1}\}$ to be the $(i+1)$-st edge of the matching.
   Let $\calR_{i+1}$ denote the ranges of $\calR$ that are crossed by the pair $\{x_{i+1},y_{i+1}\}$.
The weight of the system $\calR$ after the $(i+1)$-st step, will therefore be 
\begin{eqnarray*}
   w_{i+1}(\calR) &=& w_i(\calR) - w_i(\calR_{i+1}) + w_{i+1}(\calR_{i+1}).
\end{eqnarray*}
Now since the crossing number of the ranges in $\calR_{}$ increased by $1$ after choosing $\{x_{},y_{}\}$, therefore by the 
definition of the weight function, their weights double after the $(i+1)$-st step, that is, 
\[
    w_{i+1}(\calR_{i+1})=2w_i(\calR_{i+1}).
\]
Thus we get 
\begin{eqnarray}
   w_{i+1}(\calR) &=& w_i(\calR) - w_i(\calR_{i+1}) + 2w_{i}(\calR_{i+1})  \notag \\
                  &=& w_i(\calR)\pth{1+\frac{w_{i}(\calR_{i+1})}{w_i(\calR)}}. \label{eqn:rec-rel-wi}
\end{eqnarray}
Noting that $|\calR_{i+1}| \leq \delta$ by our choice of the pair $\{x_{i+1},y_{i+1}\}$, we get 
\begin{eqnarray*}
    w_i(\calR_{i+1}) &\leq& C\cdot \pth{w_i(\calR)^{1+c_1}t^{c}\cdot\pth{\frac{g(w_i(\calR))}{n-2i}}}^{1/(1+c_1+c)}.
\end{eqnarray*}
Now substituting in the recurrence relation~\eqref{eqn:rec-rel-wi} for the weight of $\calR$, gives 
\begin{eqnarray*}
   w_{i+1}(\calR) &\leq& w_i(\calR)\pth{1+\frac{w_i(\calR_{i+1})}{w_i(\calR)}} \\
                  &=& w_i(\calR)\pth{1+\pth{\frac{t^{c}}{(n-2i)}\cdot\frac{g(w_i(\calR))}{w_i(\calR)^{c}}}^{\frac{1}{1+c_1+c}}}. 
\end{eqnarray*}
Since by definition $g(n) = o(n^{\e})$ for every $\e\in (0,1]$ and we assumed $c$ is constant, therefore the ratio $\frac{g(w_i(\calR))}{w_i(\calR)^{c}}$ is 
maximized when $w_i(\calR)$ is the minimum, i.e. $w_0(\calR)$, or $|\calR|$. Further, by our assumption that $|\calR|\geq n$, the ratio becomes at most 
$\frac{g(n)}{n^{c}}$. Thus we get 
\begin{eqnarray*}
  w_{n/2}(\calR) &\leq& w_0(\calR)\prod_{i=1}^{n/2-1}\pth{1+\pth{\frac{t^{c}g(n))}{(n-2i)n^{c}}}^{\frac{1}{1+c_1+c}}} \\
                 &\leq& |\calR|\prod_{i=1}^{n/2-1}\pth{1+\pth{\frac{t^{c}g(n)}{(n-2i)n^{c}}}^{\frac{1}{1+c_1+c}}} .
\end{eqnarray*}

Taking logarithms with respect to both sides, we get 
\begin{eqnarray}
   \log w_{n/2}(\calR) &\leq&  \log|\calR| + O\pth{\pth{\pth{\frac{t}{n}}^{c}g(n)}^{\frac{1}{1+c_1+c}} \sum_{i=1}^{n/2-1}\frac{1}{(n-2i)^{1/(1+c_1+c)}}} \notag \\
     &\leq&  \log|\calR| + O\pth{\pth{t^{c}g(n)}^{\frac{1}{1+c_1+c}}n^{c_1/(1+c_1+c)}}  \notag \\
     &=&  \log|\calR| + O\pth{\pth{n^{c_1}t^{c}g(n)}^{\frac{1}{1+c_1+c}}}. \label{eqn:bd-max-cross-num} 
\end{eqnarray}
Here in the second inequality above, we have used the fact that 
\[
    \sum_{i=1}^{n/2-1}\frac{1}{(n-2i)^{\frac{1}{1+c_1+c}}} \leq n^{1-\frac{1}{1+c_1+c}}.
\]
Finally, let $c_{\max} := \max_{S\in \calR} c_S(n/2)$ denote the maximum number of crossings over all sets $S\in \calR$.
We have 
\[
    c_{\max} = \max_{S\in \calR}\log w_{n/2}(S) \leq \log w_{n/2}(\calR).
\]
By Equation~\eqref{eqn:bd-max-cross-num}, this is at most 
\[
    \log |\calR|+O\pth{\pth{n^{c_1}t^cg(n)}^{1/(1+c_1+c)}}.
\]
Since the VC dimension of the system is bounded (in fact, substituting $k=n$ in the 
shallow cell complexity, we see that the VC dimension is at most $2+c_1+c$), we get that $\log|\calR| = O(\log n)$, since the number of ranges 
is polynomial in the size of the universe $X$. 

This completes the proof of the theorem.
\end{proof}

It only remains to prove Corollary~\ref{cor:low-cross-match-list}. The corollary follows immediately from known bounds on the shallow cell 
complexity and the union complexity of various geometric set systems. These bounds are given in the table below. 

\begin{proof}[Proof of Corollary~\ref{cor:low-cross-match-list}]
The proof of the corollary follows by applying the shallow cell complexity functions in the table below, to the result of Theorem~\ref{thm:low-cross-match-dual}.
These bounds can be found in~\cite{DuttaGJM19} and~\cite{mustafa2022samplingbook}[see the list in Chapter 4.3] and the references therein. 
We also recommend the survey~\cite{APS08} for the interested reader. The bounds for orthants was proved in~\cite{10.5555/1283383.1283467}.

  \begin{table}[h]
    \begin{center}
    \begin{tabular}{|l|l|}
          \hline
          Objects & $m\psi(m,k)$  \\
          \hline
          Intervals in $\R$      & $m $. \\
          Half-spaces in $\R^2$      & $mk $. \\
          Homothets of a convex body in $\R^2$      & $mk$. \\
          Disks in $\R^2$      & $mk$. \\
          Pseudodisks in $\R^2$ & $mk$. \\
          $\alpha$-fat triangles in $\R^2$ & $mk\log^*\frac{m}{k}+mk\log^2\alpha$. \\
          Locally $\gamma$-fat semi-algebraic objects & $mk2^{O(\log^*m)}$. \\
          with bounded description complexity in $\R^2$ & \\
          Objects with union complexity $m\phi(m)$ & $m\phi(m/k)k$. \\
          Halfspaces in $\R^3$ & $mk$. \\
          Orthants in $\R^2$ and $\R^3$ & $mk$. \\
          \hline
    \end{tabular}
     \end{center}
     \vspace{10pt}
    \caption{Shallow-cell complexity bounds for different geometric set systems.}
  \end{table}

\end{proof}


\section{Conclusion}
\label{sec:conc}
   We have shown improved discrepancy bounds in set systems with near-linear shallow cell complexity, under the condition that each point belongs to a bounded number 
of sets. Though algorithmic considerations are not the focus of our paper, these are now briefly discussed. In order to obtain a low-discrepancy coloring, we need 
to compute a matching with low crossing number. Such algorithms have been studied by several authors e.g.~\cite{10.1145/73393.73397,Welzl1992,HarPel09,DBLP:conf/compgeom/CsikosM21}. 
For our purposes, it suffices to use the original algorithm of Welzl, which uses at most $O(n^3t)$ time in our case. We note that the faster algorithms of Csikos and 
Mustafa~\cite{DBLP:conf/compgeom/CsikosM21}
as well as Har-Peled~\cite{HarPel09} give matchings with slightly worse bounds on the crossing number (by a logarithmic factor in $n$), which results in a corresponding increase in the 
discrepancy.

\bibliographystyle{plainurl}
\bibliography{biblio}

\end{document}